\documentclass{amsart}
\usepackage{graphicx,amssymb}
\newtheorem{thm}{Theorem}[section]
\newtheorem{cor}[thm]{Corollary}

\theoremstyle{definition}

\theoremstyle{remark}

\numberwithin{equation}{section}

\newcommand{\av}[1]{\left\langle{#1}\right\rangle}
\newcommand{\mbold}[1]{\mbox{\boldmath${#1}$}}

\renewcommand{\d}{\mathrm{d}}
\def\beq{\begin{eqnarray}}
\def\eeq{\end{eqnarray}}
\def\beqs{\begin{equation}\begin{split}}
\def\eeqs{\end{split}\end{equation}}
\begin{document}

\title{Late-time behaviour of the tilted Bianchi type VI$_{-1/9}$ models}
 \author[S Hervik, R J van den Hoogen, W C Lim and A A Coley]{S Hervik$^{1}$, R J van den Hoogen$^{2, 1}$, W C Lim$^{3}$ and A A Coley$^{1}$ }%
\address{$^{1}$ Department of Mathematics \& Statistics, Dalhousie University,
Halifax, Nova Scotia,
Canada B3H 3J5}%
\address{$^2$ Department of Mathematics, Statistics and Computer Science,
St. Francis Xavier University, Antigonish, Nova Scotia, Canada B2G 2W5}
\address{$^3$ Department of Physics, Princeton University, Princeton NJ 08544, USA.}

\email{ herviks@mathstat.dal.ca, rvandenh@stfx.ca, wlim@princeton.edu, aac@mathstat.dal.ca}%
\date{\today}%
\begin{abstract}
We study tilted perfect fluid cosmological models with a constant
equation of state parameter in spatially homogeneous models of
Bianchi type VI$_{-1/9}$ using dynamical systems methods and numerical
simulations. We study models with and without vorticity, with an emphasis on their future asymptotic
evolution. We show that for models with vorticity there exists, 
in a small region of parameter space, a closed curve acting as the attractor.   
\end{abstract}

\maketitle

\section{Introduction} 

In recent papers the tilted Bianchi models
\cite{HBWII,HWV,Harnett,BHtilted,hervik,CH1,CH2,HHC,HHLC,HLim,HHLC2} have been studied using 
dynamical systems methods \cite{DS1,DS2}.  In this paper we will study the 
general tilted\footnote{The non-tilted Bianchi model of type VI$_{-1/9}$ was studied in \cite{HHW-19}.} perfect fluid
Bianchi model of type VI$_{-1/9}$, with a constant linear equation of state parameter 
$\gamma$ satisfying the causality conditions (i.e., no superluminal speed of sound), 
which was not studied in detail in \cite{HHLC}.  The
tilted models have shown a wide variety of phenomena including, for example, the existence
of closed curves
\cite{HHC} and incompleteness of the fluid congruence \cite{CollinsEllis,tilt1}.

The general irrotational perfect fluid type VI$_{-1/9}$ model is the most general  of
the irrotational Bianchi models.  In fact, the irrotational type VI$_{-1/9}$ model has
a 6-dimensional state space; just one dimension less than that of of the most general tilted type
VI$_{-1/9}$ model.  We will show that, regardless of whether we consider the general type
VI$_{-1/9}$ model with or without vorticity, the important late time asymptotes are the
non-tilted Collins type VI$_{-1/9}$ solution (for $2/3<\gamma<10/9$), \[ \d s^2=-\d
t^2+t^2\d x^2+t^{\frac{2(4-3\gamma)}{3\gamma}}e^{-\frac{2sx}{3\gamma}}\d
y^2+t^{\frac{4}{3\gamma}}e^{\frac{4sx}{3\gamma}}\d z^2, \] where
$s=\sqrt{(2-\gamma)(3\gamma-2)}$, or the Collinson-French (or Robinson-Trautman) vacuum
solution (for $10/9<\gamma<2$), \[ \d s^2=-\d t^2+t^2\d x^2+t^{\frac
25}\left(e^{-\frac{\sqrt{6}}{5}x}\d y+\frac{\sqrt{5}}2t^{\frac 45}\d
x\right)^2+t^{\frac 65}e^{\frac{4\sqrt{6}}{5}x}\d z^2.\] The tilt of the fluid can be
asymptotically non-tilted, intermediate or extreme depending on the value of $\gamma$.
The Collinson-French solution is a Petrov type III vacuum solution, and hence it
has the peculiar feature that all the curvature invariants of zeroth order
vanish\footnote{However, considering differential invariants they will not vanish; for
example, $C_{\alpha\beta\gamma\delta;\mu}C^{\alpha\beta\gamma\delta;\mu}\neq 0$.}.

The spatial hypersurfaces in a Bianchi type VI$_{-1/9}$ cosmology are defined as the
orbit of the action of the Bianchi type VI$_{-1/9}$ Lie group.  This Lie group can be
equipped with the left-invariant one-forms:  \[ \widetilde{\mbold\omega}^1=\d x, \quad
\widetilde{\mbold\omega}^2=e^{-bx}\d y,\quad \widetilde{\mbold\omega}^3=e^{2bx}\d z,\]
where $b$ is a constant.  The exceptional feature of this model does not relate to the
group itself (for which there is nothing extraordinary) but arises from the Einstein field
equations.  For this particular model one of the Einstein constraint equations vanishes
exactly and hence the vacuum case allows for an extra shear degree of freedom.  This
extra shear is present in the Collinson-French vacuum solution and therefore this metric 
has no
analogue in the other Bianchi vacuum cosmologies.

The aim of this paper is to fill one of the gaps in the analysis of the general 
tilted type VI$_h$ models presented in
\cite{HHLC}.  The $h=-1/9$ case is special (as explained above) and requires a separate
analysis, which has only been done in part previously.  In particular,  we will show
that closed periodic orbits exist for the general vortic model.  However, for the
irrotational models, no such closed curve exists.  Moreover, we will also
complete a centre manifold analysis appearing in one 6-dimensional invariant subspace
describing vortic Bianchi type VI$_{-1/9}$ models.  All other attractors will be
included for completeness.

\section{Equations of motion}
\subsection{The orthonormal frame approach}
The line-element of a Bianchi cosmology can be written
\beq
\d s^2=-\d t^2+\delta_{ab}{\mbold\omega}^a{\mbold\omega}^b,
\eeq 
where $t$ is the co-moving cosmological time. The one-forms ${\mbold\omega}^a$ are left-invariant one-forms on the hypersurfaces spanned by the group orbits and can be related to the left-invariant one-forms, $\widetilde{\mbold\omega}^i$, given above by ${\mbold\omega}^a={\sf e}^a_{~i}(t)\widetilde{\mbold\omega}^i$. 


The geometric (or normal) congruence, $n^{\mu}$, is given by ${\bf n}=\partial/\partial t$. It is also useful to define and the Hubble scalar and the shear associated with the congruence $n^\mu$: 
\beq
H\equiv\frac 13n^\mu_{~;\mu}, \quad \sigma_{\mu\nu}\equiv n_{\mu;\nu}-Hh_{\mu\nu}. 
\eeq
The matter variables are chosen to be the energy density, $\mu$, and the tilt-velocity, $v^a$, which is defined as the 3-velocity of the fluid with respect to the geometric (or normal) congruence, $n^{\mu}$. 
The equations of motion can now be written down in terms of the Hubble scalar, $H;$ the shear, $\sigma_{ab}$; the curvature variables $n^{ab}$ and $a_c$; and the matter variables $\mu$ and $v^a$.  

Expansion-normalised variables are introduced  (by dividing the variables with the appropriate powers of $H$).   
The papers \cite{CH2,HHLC2} contain all the details regarding the  determination of the evolution 
equations for the tilted cosmological models under consideration. 
Here, we shall adopt the so-called $N$-gauge in which the function 
${\bf N}_{\times}$ is purely imaginary; this is realised by the choice
$\phi'=\sqrt{3}\lambda\Sigma_-$, where $\lambda$ is defined by
$\bar{N}=\lambda\mathrm{Im}({\bf N}_{\times})$. The evolution equation for
$\bar{N}$ can then be replaced by an evolution equation for $\lambda$,
which ensures a closed system of equations.  
We will also adopt the dimensionless time parameter $\tau$, which is related to the cosmological time $t$ via $\mathrm{d} t/\mathrm{d}\tau =(1/H)$, where $H$ is the Hubble scalar. 

Using expansion-normalised variables, the equations of motion 
are (see \cite{CH2,HHLC2} for the complete derivation of the equations): 
\beq 
\Sigma_+'&=& (q-2)\Sigma_++{3}(\Sigma_{12}^2+\Sigma^2_{13})-2N^2 +\frac{\gamma\Omega}{2G_+}\left(-2v_1^2+v_2^2+v_3^2\right) \\
\Sigma_-'&=&(q-2-2\sqrt{3}\Sigma_{23}\lambda)\Sigma_-+\sqrt{3}(\Sigma_{12}^2-\Sigma_{13}^2) +2AN +\frac{\sqrt{3}\gamma\Omega}{2G_+}\left(v_2^2-v_3^2\right)\\
\Sigma'_{12}&=& \left(q-2-3\Sigma_+-\sqrt{3}\Sigma_-\right)\Sigma_{12} -\sqrt{3}\left(\Sigma_{23}+\Sigma_-\lambda\right)\Sigma_{13} +\frac{\sqrt{3}\gamma\Omega}{G_+}v_1v_2\\
\Sigma'_{13}&=&\left(q-2-3\Sigma_++\sqrt{3}\Sigma_-\right)\Sigma_{13}-\sqrt{3}\left(\Sigma_{23}-\Sigma_-\lambda\right)\Sigma_{12}+\frac{\sqrt{3}\gamma\Omega}{G_+}v_1v_3\\
\Sigma'_{23}&=&(q-2)\Sigma_{23}-2\sqrt{3}N^2\lambda+2\sqrt{3}\lambda\Sigma_-^2+2\sqrt{3}\Sigma_{12}\Sigma_{13}+ \frac{\sqrt{3}\gamma\Omega}{G_+}v_2v_3\\
N'&=& \left(q+2\Sigma_++2\sqrt{3}\Sigma_{23}\lambda\right){N}\\
\lambda' &=& 2\sqrt{3}\Sigma_{23}\left(1-\lambda^2\right)\\
A'&=& (q+2\Sigma_+)A . 
\eeq 
The equations for the fluid are:
\beq
\Omega'&=& \frac{\Omega}{G_+}\Big\{2q-(3\gamma-2)+2\gamma Av_1
 +\left[2q(\gamma-1)-(2-\gamma)-\gamma\mathcal{S}\right]V^2\Big\}
 \quad \\
 v_1' &=& \left(T+2\Sigma_+\right)v_1-2\sqrt{3}\Sigma_{13}v_3-2\sqrt{3}\Sigma_{12}v_2-A\left(v_2^2+v_3^2\right)-\sqrt{3}N\left(v_2^2-v_3^2\right)\\
 v_2'&=& \left(T-\Sigma_+-\sqrt{3}\Sigma_-\right)v_2-\sqrt{3}\left(\Sigma_{23}+\Sigma_-\lambda\right)v_3+\sqrt{3}\lambda{N}v_1v_3+\left(A+\sqrt{3}N\right)v_1v_2 \\
 v_3'&=& \left(T-\Sigma_++\sqrt{3}\Sigma_-\right)v_3-\sqrt{3}\left(\Sigma_{23}-\Sigma_-\lambda\right)v_2-\sqrt{3}\lambda{N}v_1v_2+\left(A-\sqrt{3}N\right)v_1v_3 \\
 V'&=&\frac{V(1-V^2)}{1-(\gamma-1)V^2}\left[(3\gamma-4)-2(\gamma-1)Av_1-\mathcal{S}\right],
\eeq 
where 
\beq q&=& 2\Sigma^2+\frac
12\frac{(3\gamma-2)+(2-\gamma)V^2}{1+(\gamma-1)V^2}\Omega\nonumber \\
\Sigma^2 &=& \Sigma_+^2+\Sigma_-^2+\Sigma_{12}^2+ \Sigma_{13}^2+\Sigma_{23}^2\nonumber \\
\mathcal{S} &=& \Sigma_{ab}c^ac^b, \quad c^ac_{a}=1, \quad v^a=Vc^a,\quad \nonumber \\
 V^2 &=& v_1^2+v_2^2+v_3^2,\quad  \nonumber \\
 T&=& \frac{\left[(3\gamma-4)-2(\gamma-1)Av_1\right](1-V^2)+(2-\gamma)V^2\mathcal{S}}{1-(\gamma-1)V^2}\nonumber\\
 G_+&=&1+(\gamma-1)V^2\nonumber.
\eeq 
These variables are subject to the constraints 
\beq
1&=& \Sigma^2+A^2+N^2+\Omega \label{const:H}\\
0 &=& 2\Sigma_+A+2\Sigma_-N+\frac{\gamma\Omega v_1}{G_+} \label{const:v1}\\
0 &=&
-\left[\Sigma_{12}(N+\sqrt{3}A)+\Sigma_{13}\lambda{N}\right]+\frac{\gamma\Omega v_2}{G_+} \label{const:v2}\\
0 &=&
\left[\Sigma_{13}(N-\sqrt{3}A)+\Sigma_{12}\lambda{N}\right]+\frac{\gamma\Omega v_3}{G_+} \label{const:v3} \\
0&=& 3A^2-\left(1-\lambda^2\right)N^2.\label{const:group} 
\eeq 
The parameter $\gamma$ will be assumed to satisfy $\gamma\in (0,2)$. 
The generalized Friedmann equation (\ref{const:H}) yields an expression which effectively defines the energy density $\Omega$. We will assume that this energy density is non-negative: $\Omega\geq 0$. 
Therefore, the state vector can be considered 
${\sf X}=[\Sigma_+,\Sigma_-,\Sigma_{12},\Sigma_{13},\Sigma_{23},N,\lambda,A,v_1,v_2,v_3]$ 
modulo the constraint equations  (\ref{const:v1})-(\ref{const:group}). 
Thus the dimension of the physical state space is seven 
(and hence of equal generality to the other models of type VI$_h$ for a given value of $h$).  
Additional details are presented in \cite{CH2}.

The dynamical system is invariant under the following discrete symmetries :
$$\begin{tabular}{l}
$\phi_1:~[\Sigma_+,\Sigma_-,\Sigma_{12},\Sigma_{13},\Sigma_{23},N,\lambda,A,v_1,v_2,v_3] 
\mapsto  [\Sigma_+,\Sigma_-,\Sigma_{12},\Sigma_{13},\Sigma_{23},-N,\lambda,-A,-v_1,-v_2,-v_3] $ \\
$\phi_2:~[\Sigma_+,\Sigma_-,\Sigma_{12},\Sigma_{13},\Sigma_{23},N,\lambda,A,v_1,v_2,v_3] 
\mapsto [\Sigma_+,-\Sigma_-,\Sigma_{13},\Sigma_{12},\Sigma_{23},-N,\lambda,A,v_1,v_3,v_2] $ \\
$\phi_3^{\pm}\! :~[\Sigma_+,\Sigma_-,\Sigma_{12},\Sigma_{13},\Sigma_{23},N,\lambda,A,v_1,v_2,v_3]
\mapsto [\Sigma_+,\Sigma_-,\pm \Sigma_{12},\mp \Sigma_{13},-\Sigma_{23},N,-\lambda,A,v_1,\pm v_2,\mp v_3]$\\
$\phi_4:~[\Sigma_+,\Sigma_-,\Sigma_{12},\Sigma_{13},\Sigma_{23},N,\lambda,A,v_1,v_2,v_3]
\mapsto[\Sigma_+,\Sigma_-,-\Sigma_{12},-\Sigma_{13},\Sigma_{23},N,\lambda,A,v_1,-v_2,-v_3]$
\end{tabular}$$
These discrete symmetries imply that without loss of  generality we can restrict the 
variables $A\geq0$ and $N\geq0$,  since the dynamics in the other regions can be 
obtained by simply applying one or more of  the maps above.  The third and fourth symmetries listed imply that one can add  additional constraints on the variables $\Sigma_{12},\Sigma_{13},v_2$ or $v_3$. For the type VI$_{-1/9}$ we can use $\phi_4$ to restrict $v_2$ to be non-negative (since $v_2=0$ implies $v_3=0$, see below); hence, we will assume $v_2\geq 0$. There is no natural way to restrict any of the remaining variables using the symmetry $\phi^+_3$, and so we will not do so here.  

\subsection{Invariant sets}
For the case $h=-1/9$ we define ($N_{ab}v^av^b$ is identically zero for $h=-1/9$)
\beq 
\widehat{D}=\left[\lambda(\Sigma_{12}^2+\Sigma_{13}^2)+2\Sigma_{12}\Sigma_{13}\right]. 
\label{Dhat}\eeq
In this analysis we will be concerned with the following invariant sets: 
\begin{enumerate}
\item{} $T(VI_{-1/9})$: The general tilted type VI$_{-1/9}$  model: $ |\lambda|<1$. 
\item{} $T_1(VI_{-1/9})$: A one-tilted type VI$_{-1/9}$ model: $|\lambda|<1$, $v_2=v_3=\Sigma_{12}=\Sigma_{13}=0$. 
\item{} $T_{1,0}(VI_{-1/9})$: A one-tilted diagonal type VI$_{-1/9}$ model: $v_2=v_3=\Sigma_{12}=\Sigma_{13}=\Sigma_{23}=\lambda=0$.
\item{} $N^{\pm}(VI_{-1/9})$: A class of tilted type VI$_{-1/9}$ models:  $|\lambda|<1$, $\widehat{D}=0$. 
\item{} $T_2^+(VI_{-1/9})$: A two-tilted type VI$_{-1/9}$ model. This is the fixed-point-set of $\phi_3^+$ and is given by  $v_3=\Sigma_{13}=\Sigma_{23}=\lambda=0$.
\item{} $T_2^-(VI_{-1/9})$: A one-tilted irrotational type VI$_{-1/9}$ model. This is the fixed-point-set of $\phi^-_3$ and is given by $v_2=v_3=\Sigma_{12}=\Sigma_{23}=\lambda=0$.
\item{} $B(VI_{-1/9})$: Non-tilted type VI$_{-1/9}$: $|\lambda|<1$, $v_1=v_2=v_3=\Sigma_{12}=\Sigma_{13}=0$.
\item{} $B_0(VI_{-1/9})$: A class of diagonal non-tilted type VI$_{-1/9}$ models ($n^{\alpha}_{~\alpha}=0$): $V=\Sigma_{12}=\Sigma_{13}=\Sigma_{23}=\lambda=0$
\item{} $T(II)$: The general type II model: $\lambda=\pm 1$, $A=0$. 
\item{}$B(I)$: Type I: ${N}=A=V=0$.
\item{}{$\partial T(I)$}: ``Tilted'' vacuum type I: $\Omega=N=A=0$.
\end{enumerate}
Regarding $N^{\pm}(VI_{-1/9})$, to verify that this is indeed an invariant subspace, we calculate $\widehat{D}'$:
\[ \widehat{D}'=2\left(q-2-3\Sigma_+-\sqrt{3}\lambda\Sigma_{23}+3Av_1\right) \widehat{D};\] 
hence, $\widehat{D}=0$ defines an invariant subspace. 
Note also that this invariant set is not a manifold; it is similar to the light-cone in 2-dimensional Minkowski space. Therefore, $N^{\pm}(VI_{-1/9})-T_1(VI_{-1/9})$ consists of 4 disconnected pieces. By the symmetry $\phi_4$, these are actually only two inequivalent pieces. Here, we choose $N^{\pm}(VI_{-1/9})$ such that 
\[ T_2^+(VI_{-1/9})\subset N^+(VI_{-1/9}), \quad T_2^-(VI_{-1/9})\subset N^-(VI_{-1/9}) \]  
Since $N^+(VI_{-1/9})\cap N^-(VI_{-1/9})=T_1(VI_{-1/9})$, both $N^+(VI_{-1/9})$ and $N^-(VI_{-1/9})$ are invariant sets. 

We note that the closure of the set $T(VI_{-1/9})$ is given by
\beq
\overline{T(VI_{-1/9})}&=&T(VI_{-1/9})\cup T(II)\cup B(I) \cup \partial T(I).
\label{eq:decomp}\eeq
Since the boundaries may play an important role in the evolution of the dynamical system  we must consider all of the sets in the decomposition (\ref{eq:decomp}).

Let us consider the constraint equations (\ref{const:v2}) and (\ref{const:v3}) as a
linear map \[ {\sf L}:~(\Sigma_{12},\Sigma_{13})\mapsto (v_2,v_3)/G_+,\] where ${\sf
L}$ is considered as given in terms of $A$, $N$, $\lambda$ and $\Omega$.  For $h\neq
-1/9$, $\det({\sf L})\neq 0$ and the image of ${\sf L}$ is 2-dimensional.  However, for
$h= -1/9$, $\det({\sf L})= 0$ and the image of ${\sf L}$ is 1-dimensional; hence, in
this sense \emph{the constraint equations are degenerate}.  This implies that
$(v_2,v_3)$ has to be restricted to a 1-dimensional submanifold.  We will therefore say
that the general type VI$_{-1/9}$ model only allows for 2 tilt degrees of freedom.  In
particular, we can solve for $v_3$ and obtain \[ v_3=-\frac{\lambda
v_2}{1+\sqrt{1-\lambda^2}}.\]

It is illustrative to consider the eigenvectors of the map ${\sf L}$:
\beq
{\bf x}_0=\left(
-\frac{\lambda}{1+\sqrt{1-\lambda^2}}, ~ 1
\right),&& \quad {\sf L}{\bf x}_0=0, \nonumber \\
{\bf x}_a=\left(
1,~ -\frac{\lambda}{1+\sqrt{1-\lambda^2}} 
\right),&& \quad {\sf L}{\bf x}_a=a{\bf x}_a, \quad a>0. \nonumber
\eeq
For each of these eigenvectors, $\widehat{D}=0$ and hence, alternatively, we can define $N^{\pm}(VI_{-1/9})$  when $(\Sigma_{12},\Sigma_{13})$ is proportional to one of these eigenvectors. More specifically, for $N^+(VI_{-1/9})$, $(\Sigma_{12},\Sigma_{13})\propto {\bf x}_a$, while for $N^-(VI_{-1/9})$, $(\Sigma_{12},\Sigma_{13})\propto {\bf x}_0$. 
 
Note also that, due to the presence of the eigenvector with zero eigenvalue ${\bf x}_0$, $(v_2,v_3)=0$ does not necessarily mean that $(\Sigma_{12},\Sigma_{13})$ is zero. In particular, for the non-tilted models this implies that we may have an additional shear degree of freedom; these models have usually been called the exceptional case and are denoted by an asterisk; e.g., $B_0(VI^*_{-1/9})$ and  $B(VI^*_{-1/9})$. We also note that for the tilted models, there is an exceptional case of the one-tilted models $T_1(VI_{-1/9})$ which could be denoted $T_1(VI^*_{-1/9})$. However, $T_1(VI^*_{-1/9})=N^-(VI_{-1/9})$ as explained above. Therefore, we keep the notation $N^-(VI_{-1/9})$. Similarly, we have $T_{1,0}(VI^*_{-1/9})=T_2^-(VI_{-1/9})$.

\begin{table}
\caption{The dimensions of the invariant sets for the Bianchi type $VI_{ -1/9}$
model. The right-most column indicates the specialization in terms of the $G_2$ cosmologies. Here, Diag means diagonal, HO means hypersurface orthogonal, and OT means orthogonally transitive. The stars indicate the exceptional cases where the models aquire an addition shear degree of freedom.}
\begin{tabular}{|c||cc|c|c|}
\hline
Dim &Invariant set & &No of Tilt & $G_2$ action \\
\hline 
2 & $B_0(VI_{-1/9})$ &&0  & Diag \\ \hline
3 & $T_{1,0}(VI_{-1/9})$ & &1 & Diag \\ 
  & $B_0(VI^*_{-1/9})$& $\star$ &0 & HO \\
\hline
4 & $B(VI_{-1/9})$ & &0  & OT \\
  & $T_2^+(VI_{-1/9})$ & &2  & HO \\
  & $T_2^-(VI_{-1/9})$ & $\star$ &1  & HO \\
\hline
5 & $T_1(VI_{-1/9})$ & &1  & OT \\
  & $B(VI^*_{-1/9})$& $\star$ &0  & Gen $G_2$ \\
\hline 
6 & $N^+(VI_{-1/9})$ & & 2  &Gen $G_2$ \\
  & $N^-(VI_{-1/9})$ &$\star$ & 1  &Gen $G_2$ \\
\hline
7 & $T(VI_{-1/9})$ & &2 & Gen $G_2$ \\
\hline 
\end{tabular}\\
\end{table}
\subsection{Fluid Vorticity} 
The various invariant subspaces can also be categorised in terms of
the ($H_{\mathrm{fluid}}$-normalised, where $H_{\mathrm{fluid}}\equiv (1/3)u^{\mu}_{~;\mu}$) fluid vorticity, $W^{\alpha}$. The vorticity of the fluid for the type VI$_{-1/9}$ models is given by:
\beq
W_a=\frac{1}{2B}\left(N_{ab}+\varepsilon_{abc}A^c\right)v^b, \quad W_0=0, 
\label{eq:vorticity}\eeq
where 
\[ B\equiv\frac{1-\frac 13(V^2+V^2\mathcal{S}+2A_av^a)}{[1-(\gamma-1)V^2]\sqrt{1-V^2}}.\]
For the invariant sets: 
\begin{enumerate}
\item{} $T(VI_{-1/9})$: $W^0=W^1=0$, most general vortic type VI$_{-1/9}$.
\item{} $N^+(VI_{-1/9})$: $W^0=W^1=0$.
\item{} $N^-(VI_{-1/9})$: $W^0=W^a=0$, non-vortic.
\item{} $T_2^+(VI_{-1/9})$: $W^0=W^1=W^2=0$.
\item{} $T_2^-(VI_{-1/9})$: $W^0=W^a=0$, non-vortic.
\item{} $T_1(VI_{-1/9})$: $W^0=W^a=0$, non-vortic.
\item{} $B(VI_{-1/9})$: $W^0=W^a=0$, non-tilted and non-vortic.
\end{enumerate}
We note that the most general non-vortic model is of dimension 6. Hence, since the non-vortic type VI$_h$ and VII$_h$ models -- regarding $h$ as fixed -- are of dimension 5, the type VI$_{-1/9}$ model is the most general non-vortic model of all Bianchi models. 

In general we can also solve for the vorticity component $W^2$: 
\[ W^2=\frac{\lambda}{1+\sqrt{1-\lambda^2}}W^3.\] 
This follows from the constraint equations and eq.(\ref{eq:vorticity}).

\section{Qualitative behaviour}
\label{sect:Qual}
\subsection{Monotone functions} 
There are a number of monotone functions in the state space of interest. For $0<\gamma\leq 6/7$, there exists a monotonically increasing function $Z_1$ defined by
\beq
Z_1&\equiv& \alpha\Omega^{1-\gamma}, \quad \alpha=\frac{(1-V^2)^{\frac 12(2-\gamma)}}{G^{1-\gamma}_+V^{\gamma}_{\phantom{+}}}, \\
Z_1'&=&\left[2(1-\gamma)q+(2-\gamma)+{\gamma}\mathcal{S}\right]Z_1.\nonumber 
\eeq
This function can be used to show \cite{CH2}:
\begin{thm}\label{thm:nontilted}
For $0<\gamma\leq 6/7$, all tilted Bianchi models  (with $\Omega>0$, $V<1$) of solvable type are asymptotically non-tilted at late times. 
\end{thm}
\begin{cor}[Cosmic no-hair]
For $\Omega>0$, $V<1$, and $0<\gamma<2/3$ we have that
\[ \lim_{\tau\rightarrow \infty}\Omega=1, \quad \lim_{\tau\rightarrow \infty}V=0.\]
\label{no-hair}
\end{cor}

Moreover, the following function is a monotone function in $T(VI_{-1/9})$:
\beq
Z_2&=&\frac{A^4N^2G_+^5\widehat{D}^2}{(1-V^2)^{\frac 52(2-\gamma)}\Omega^5}, \\
Z_2'&=&(5\gamma-6)(3-2Av_1)Z_2, 
\eeq
where $\widehat{D}$ is defined in eq.(\ref{Dhat}). 
This function is monotonically decreasing for $\gamma<6/5$ and monotonically increasing for $6/5<\gamma$. 

We note that in the subspace $N^-(VI_{-1/9})$ we have the monotone function: 
\beq
Z_3&=&\frac{v_1^2\Omega}{A^2G_+(1-V^2)^{\frac 12(2-\gamma)}}, \\
Z_3'&=&-(2-\gamma)(3-2Av_1)Z_3.
\eeq
This function is monotonically decreasing in $N^-(VI_{-1/9})$. 

The monotonic function $Z_3$ immediately implies:
\begin{thm}[Future behaviour in $N^-(VI_{-1/9})$]
For $2/3<\gamma<2$, $\Omega>0$, $A>0$, $v_1^2<1$, $v_2=v_3=0$ we have that:  
\[ \text{either}\quad \lim_{\tau\rightarrow \infty}\Omega=0, \quad
\text{or}\quad  \lim_{\tau\rightarrow \infty}V=0.\]
\end{thm}
This implies that all irrotational  type VI$_{-1/9}$ universes are either asymptotically vacuum or non-tilted at late times. 
\subsection{Equilibrium points} 

\subsubsection{$B(I)$: Equilibrium points of Bianchi type I}  
\begin{enumerate} 
\item{}$\mathcal{I}(I)$: $\Sigma_+=\Sigma_-=\Sigma_{12}=\Sigma_{13}=\Sigma_{23}=A=N=V=0$ and $\Omega=1$. Here, $|\lambda|<1$ and is an unphysical parameter.  This represents the flat Friedman-Lema{\^i}tre model. 
\end{enumerate}
The remaining equilibrium points are all in $\partial T(I)$.

\subsubsection{$T(II)$: Equilibrium points of Bianchi type II} 
All of the tilted equilibrium points come in pairs. These represent identical solutions (they differ by a frame rotation); however, since their embeddings in the full state space are inequivalent, two of their eigenvalues are different. All of these equilibrium points have an unstable direction with eigenvalue $-2\sqrt{3}\Sigma_{23}$ corresponding to the variable $A$. 
These equilibrium points are given in \cite{HHC}. 

\subsubsection{$T(VI_{-1/9})$:  Equilibrium points of Bianchi type VI$_{-1/9}$}
\begin{enumerate} \item{} $\mathcal{C}(-1/9)$:  Collins perfect fluid solution,
$2/3<\gamma<5/3$ \\ $\Sigma_{12}=\Sigma_{13}=\Sigma_{23}=\lambda=V=0$, $\Sigma_+=-\frac
1{4}(3\gamma-2)$, $\Sigma_-=\frac{\sqrt{3}}{12}(3\gamma-2)$, $N^2=\frac
3{16}(3\gamma-2)(2-\gamma)$, $A=N/\sqrt{3}$, $\Omega=\frac 13(5-3\gamma)$.  This
equilibrium point is in $B(VI_{-1/9})$.  \item{} $\mathcal{R}^+(-1/9)$:  The
Apostolopoulos $h=-1/9$ solution \cite{Apo2}, $4/3<\gamma<3/2$\\ This is a vortic
solution lying in the invariant subspace $T^+(VI_{-1/9})$.  The solution is given in
terms of the expansion-normalised variables in \cite{HHLC2} with $h=-1/9$, $k=1/3$.
\item{} Bianchi type VI$_{-1/9}$ vacuum plane waves.  All of these solutions have \[
\Omega=\Sigma_{12}=\Sigma_{13}=\Sigma_{23}=0, ~\Sigma_-=N=\sqrt{-\Sigma_+(1+\Sigma_+)},
~A=(1+\Sigma_+),~ -1<\Sigma_+<-3/4,~ |\lambda|<1.\] It is avantageous to introduce
$r\equiv\sqrt{1-\lambda^2}$, which implies that we can write \[
\Sigma_+=-\frac{3}{3+r^2}, \quad 0<r\leq 1.  \] We will also define $\rho$ by \[
\rho=v_2^2+v_3^2.\]

The equilibrium points are then determined by the tilt velocities:  \begin{enumerate}
\item{} $\mathcal{L}(-1/9)$:  $v_1=v_2=v_3=0$.  These represent 'non-tilted' plane
waves and lie in $B(VI_{-1/9})$.  \item{} $\widetilde{\mathcal{L}}(-1/9)$:
$v_1=\frac{3\gamma(3+r^2)-2(9+2r^2)}{2r^2(\gamma-1)}$, $v_2=v_3=0$,
$\frac{6(3+r^2)}{9+5r^2}<\gamma<2$.  These represent 'intermediately tilted' plane
waves and lie in $T_1(VI_{-1/9})$.  \item{} $\widetilde{\mathcal{L}}_{\pm}(-1/9)$:
$v_1=\pm 1$, $v_2=v_3=0$.  These represent 'extremely tilted' plane waves and lie in
$T_1(VI_{-1/9})$.  \item{} \label{defF} $\widetilde{\mathcal{F}}^+(h)$:  Here,
$(9+7r^2)/[3(3+r^2)]\leq \gamma\leq 3(3+r^2)/(9-r^2)$ and \beq
v_1&=&-\frac{3\gamma(3+r^2)-(9+7r^2)}{2r^2(3-\gamma)},\quad v_2^2-v_3^2= \rho r
\nonumber \\
\rho&=&\frac{(9+r^2)\left[5-3\gamma\right]\left[3\gamma(3+r^2)-(9+7r^2)\right]}{8r^4(3-\gamma)^2}
\nonumber \eeq These represent 'intermediately tilted' plane waves and lie in
$N^+(VI_{-1/9})$ (for $\lambda=0$ they lie in $T_2^+(VI_{-1/9})$).  \item{}
$\widetilde{\mathcal{E}}_p^+(-1/9)$, $0<\gamma<2$:  \beq
v_1&=&-\frac{4r^2}{3(3-r^2)},\quad v_2^2-v_3^2=\rho r \nonumber \\ \rho&=&
1-v_1^2=\frac{(9+r^2)(9-7r^2)}{9(3-r^2)^2} \nonumber \eeq These represent 'extremely
tilted' plane waves and lie in $N^+(VI_{-1/9})$.  For $\lambda=0$ these equilibrium
points lie in $T_2^+(VI_{-1/9})$, and due to the special importance for the late-time
behaviour, we will denote $\widetilde{\mathcal{E}}_{p0}^+(-1/9)\equiv
\left.\widetilde{\mathcal{E}}_p^+(-1/9)\right|_{\lambda=0}$.  Therefore, \[
\widetilde{\mathcal{E}}_{p0}^+(-1/9):  ~ v_1=-\frac 23, \quad v_2=\frac{\sqrt{5}}{3}.\]
\end{enumerate} \item{} $\mathcal{CF}$:  The Collinson-French (Robinson-Trautmann)
solution is given by:  \beq &&\Sigma_+=-\frac 13,\quad \Sigma_-=\frac{1}{3\sqrt{3}},
\quad \Sigma_{13}=\frac{\sqrt{15}}{9},\quad N=\frac{1}{\sqrt{2}}, \quad
A=\frac{1}{\sqrt{6}},\nonumber \\ && \Sigma_{12}=\Sigma_{23}=\Omega=\lambda=0.
\nonumber \eeq There are the following equilibrium points associated with the
Collinson-French solution:  \begin{enumerate} \item{}$\mathcal{CF}_0$:  $v_1=v_2=0$,
$0<\gamma<2$.  \item{}$\widetilde{\mathcal{CF}}_{1+}$:
$v_1=-\frac{\sqrt{6}(3\gamma-4)}{2(3-\gamma)}$,
$v_2=\frac{\sqrt{5(3\gamma-4)(3-2\gamma)}}{\sqrt{2}(3-\gamma)}$, $\frac 43<\gamma<\frac
32$.  \item{} $\widetilde{\mathcal{CF}}_{2}$:
$v_1=\frac{\sqrt{6}(9\gamma-14)}{6(\gamma-1)}$, $v_2=0$,
$\frac{24-\sqrt{6}}{15}<\gamma<\frac{24+\sqrt{6}}{15}$.
\item{}$\widetilde{\mathcal{ECF}}_{\pm}$:  $v_1=\pm 1$, $v_2=0$, $0<\gamma<2$.
\end{enumerate} These equilibrium points, and their stability, were studied in
\cite{CH2}.

\item{} $\mathcal{W}$: Wainwright $\gamma=10/9$ solution:
\beq
&&\Sigma_+=-\frac 13,\quad \Sigma_-=\frac{1}{3\sqrt{3}},\quad 0<\Sigma_{13}<\frac{\sqrt{15}}{9},\quad
N=\frac{1}{6}\sqrt{8+54\Sigma_{13}^2}, \quad A=\frac{1}{\sqrt{3}}N,\quad \Omega=\frac 59-3\Sigma_{13}^2,\nonumber \\
&& \Sigma_{12}=\Sigma_{23}=\lambda=V=0.  \nonumber
\eeq 
\end{enumerate}

\section{Late-time behaviour} 

The late-time behaviour of models with $0<\gamma<2/3$ is determined by
(the Cosmic no-hair) Corollary 3.2.

\subsection{The invariant subspace $N^+(VI_{-1/9})$}
Here, we have the following late-time attractors:
\begin{itemize}
\item{} $2/3<\gamma\leq 4/3$: The Collins solution, $\mathcal{C}(-1/9)$. 
\item{} $4/3<\gamma<3/2$: The Apostolopoulos solution, $\mathcal{R}^+(-1/9)$. 
\item{} $3/2\leq \gamma<2$: "Extremely tilted" vacuum plane waves, $\widetilde{\mathcal{E}}^+_{p0}({-1/9})$ 
\end{itemize}
The stability of these points for $\gamma<3/2$ follows from the eigenvalues of the linearised matrix. For $3/2\leq \gamma<2$  several zero-eigenvalues occur and a centre manifold analysis is needed to determine 
the late-time behaviour. 
\subsubsection{The case $3/2 <\gamma<2$: the centre manifold}   Let us present the 
centre manifold analysis of the equilibrium point $\widetilde{\mathcal{E}}^+_{p0}({-1/9})$  in some detail.
The centre manifold in this case is a 2-dimensional submanifold of the 5-dimensional extremely tilted invariant subspace $\left.N^+(VI_{-1/9})\right|_{V=1}$. To find the centre manifold, we will therefore set $V=1$. Let us choose variables 
\beq
(\Sigma_+,\Sigma_{23},N,\lambda,v_2)=\left(-\frac 34+x_1,x_2,\frac{\sqrt{3}}{4}+x_3,x_4,\frac{\sqrt{5}}3+x_5\right),
\eeq
let $\Sigma_-$, $\Sigma_{12}$, $\Sigma_{13}$ and $\Omega$ be determined from the constraint equations, and $v_1^2=1-v_2^2-v_3^2$. Let us define the column vector ${\sf x}=[x_1,x_2,x_3,x_4,x_5]^T$. We can now expand the equations of motion to 2nd order in ${\sf x}$:
\beq
{\sf x}'={\sf A}{\sf x}+{\sf C}({\sf x},{\sf x})+\mathcal{O}({\sf x}^3),
\eeq
where ${\sf C}(-,-)$ is a bilinear vector-valued function. 
It is convenient to align the vector ${\sf x}$ with the Jordan canonical form of ${\sf A}$. This can be accomplished by defining 
\beq
{\sf P}=\begin{bmatrix} 
\frac 75 & -\frac 25 & 0 & 0 & 0 \\
0 & 1 & 0 & 0 & 1 \\
\frac{2\sqrt{3}}5 & -\frac{2\sqrt{3}}5 & 0 & 0 & 0 \\
4\sqrt{3} & -4\sqrt{3} & 0 & 4\sqrt{3} & -4\sqrt{3} \\
-\frac{48\sqrt{5}}{25} & \frac{64\sqrt{5}}{45} & \frac{112\sqrt{5}}{225} & 0 & 0
\end{bmatrix}.
\eeq
Now, defining ${\sf y}={\sf P}^{-1}{\sf x}$, the equation for ${\sf y}$ becomes
\[ 
{\sf y}'={\sf J}{\sf y}+\tilde{{\sf C}}({\sf y},{\sf y})+\mathcal{O}({\sf y}^3),
\]
where ${\sf J}$ is the Jordan block matrix 
\[ {\sf J}=\mathrm{diag}\left(0,-\frac 12, -\frac 56, 0, -\frac 12\right).\]
The centre manifold correspond to the zero-eigenvalues of ${\sf J}$. 
We can therefore parameterise the centre manifold using the variables $(y_1,y_4)$. 
The next step is to expand the variables $y_i$, $i=2,3,5$, in terms of $(y_1,y_4)$ on 
the centre manifold, to second order. We therefore define the quadratic forms $Y_i(y_1,y_4)$, $i=2,3,5$, such that $y_i-Y_i(y_1,y_4)=0$ is (to second order) an invariant subspace. On the centre manifold we then have: 
\beq
y_i=Y_i(y_1,y_4)+\mathcal{O}({\sf y}^3). 
\eeq
By substituting these into the $y_1$ and $y_4$ evolution equations, we finally get, on the centre manifold: 
\beq
y_1'&=&-16y_1^2+\mathcal{O}({\sf y}^3), \nonumber \\
y_4'&=&-16y_1y_4+\mathcal{O}({\sf y}^3).
\eeq
To lowest order, these equations can be solved to give $y_1\approx 1/(16\tau)$,
$y_4\approx C/\tau$, where a constant of integration has been eliminated by a
translation of time.  Note that, in principle, there are essentially two kinds of behaviour:  for
$\tau<0$ the variables diverge, while for $\tau>0$ the variables decay.
However, requiring $\Omega>0$ leaves only the decaying mode as physically acceptable.
These decay rates will therefore be the dominant ones since the centre manifold will
dominate the behaviour at late times.

In terms of the original variables, we therefore get the decay rates in $N^+(VI_{-1/9})$ (only dominant decay rates are included): 
\beq
\Sigma_+&\approx & -\frac 34\left(1-\frac{7}{60\tau}\right), \nonumber \\
\Sigma_-&\approx & \frac{\sqrt{3}}4\left(1+\frac{1}{60\tau}\right), \nonumber \\
\Sigma_{12}&\approx & \frac{\sqrt{15}}{60\tau}, \nonumber \\
\Sigma_{13}&\approx & -\frac{(1+16C)\sqrt{5}}{160\tau^2}, \nonumber \\
\Sigma_{23}&\approx & -\frac{(1+16C)}{8\tau^2}, \nonumber \\
N&\approx & \frac{\sqrt{3}}4\left(1+\frac{1}{10\tau}\right), \nonumber \\
\lambda &\approx & \frac{(1+16C)\sqrt{3}}{4\tau}, \nonumber \\
\Omega &\approx & \frac{3}{40\tau}, \nonumber \\
v_1&\approx & -\frac{2}3\left(1+\frac{9}{20\tau}\right), \nonumber \\
v_2&\approx & \frac{\sqrt{5}}3\left(1-\frac{9}{25\tau}\right), \nonumber \\
\sqrt{1-V^2}&\approx & (\sqrt{1-V^2})_0\exp\left[-\frac{5(2\gamma-3)}{3(2-\gamma)}\tau\right].
\eeq
This confirms that the 'extremely tilted' vacuum plane wave  $\widetilde{\mathcal{E}}^+_{p0}(-1/9)$ is the attractor for $N^+(VI_{-1/9})$, $3/2<\gamma<2$. Note, however, that this solution is unstable in the fully tilted space $T(VI_{-1/9})$, which can be verified by calculating, for example, $\widehat{D}'\approx (3/4)\widehat{D}$ close to $\widetilde{\mathcal{E}}^+_{p0}(-1/9)$. 

\subsection{The invariant subspace $N^-(VI_{-1/9})$: non-vortic type VI$_{-1/9}$ models}
\begin{itemize}
\item{} $2/3<\gamma<10/9$: The Collins solution, $\mathcal{C}(-1/9)$. 
\item{} $\gamma=10/9$: The Wainwright solution, $\mathcal{W}$. 
\item{} $10/9<\gamma<2$: The Collinson-French solution, 
$\widetilde{\mathcal{C}\mathcal{F}}$. We have the following refinement for the 
asymptotic tilt (which follows from an analysis of the eigenvalues): 
\begin{itemize}
\item[$\star$] $10/9<\gamma\leq (24-\sqrt{6})/15$: The tilt is asymptotically zero ($\mathcal{C}\mathcal{F}_0$).
\item[$\star$] $(24-\sqrt{6})/15<\gamma\leq 14/9$: The tilt is either asymptotically zero ($\mathcal{C}\mathcal{F}_0$) or extreme ($\widetilde{\mathcal{E}\mathcal{C}\mathcal{F}}_-$).
\item[$\star$] $14/9<\gamma<(24+\sqrt{6})/15$: The tilt is either asymptotically intermediate ($\widetilde{\mathcal{C}\mathcal{F}}_2$) or extreme ($\widetilde{\mathcal{E}\mathcal{C}\mathcal{F}}_-$).
\item[$\star$] $(24+\sqrt{6})/15\leq \gamma<2$: The tilt is asymptotically extreme 
($\widetilde{\mathcal{E}\mathcal{C}\mathcal{F}}_-$).
\end{itemize}
\end{itemize}

\subsection{$T(VI_{-1/9})$: The fully tilted type VI$_{-1/9}$ model}
For the fully tilted model, the late-time behaviour can be summarised as follows: 
\begin{itemize}
\item{} $2/3<\gamma<10/9$: The Collins solution, $\mathcal{C}(-1/9)$. 
\item{} $\gamma=10/9$: The Wainwright solution, $\mathcal{W}$. 
\item{} $10/9<\gamma<2$: The Collinson-French solution, $\widetilde{\mathcal{C}\mathcal{F}}$. We have the following refinement for the asymptotic tilt:
\begin{itemize}
\item[$\star$] $10/9<\gamma\leq (24-\sqrt{6})/15$: The tilt is asymptotically zero  ($\mathcal{C}\mathcal{F}_0$).
\item[$\star$] $(24-\sqrt{6})/15<\gamma\leq 4/3$: The tilt is either asymptotically zero  ($\mathcal{C}\mathcal{F}_0$) or extreme ($\widetilde{\mathcal{E}\mathcal{C}\mathcal{F}}_-$).
\item[$\star$] $4/3<\gamma\leq 5/6+\sqrt{721}/42$: The tilt is either asymptotically intermediate ($\widetilde{\mathcal{C}\mathcal{F}}_{1+}$)  or extreme ($\widetilde{\mathcal{E}\mathcal{C}\mathcal{F}}_-$).
\item[$\star$] $5/6+\sqrt{721}/42<\gamma<\gamma_H\approx 1.47392$: The tilt 
is either asymptotically oscillatory (closed curve) or extreme ($\widetilde{\mathcal{E}\mathcal{C}\mathcal{F}}_-$).
\item[$\star$] $\gamma_H<\gamma<2$: The tilt is asymptotically extreme ($\widetilde{\mathcal{E}\mathcal{C}\mathcal{F}}_-$). 
\end{itemize}
\end{itemize}
Most of these results comes from an analysis of the eigenvalues of the various 
equilibrium points. However, for $\gamma= 5/6+\sqrt{721}/42$ the equilibrium point 
$\widetilde{\mathcal{C}\mathcal{F}}_{1+}$ undergoes a Hopf-bifurcation and 
a stable closed curve results. The analysis of this Hopf-bifurcation is given below. Interestingly, this closed curve co-exists with an extremely tilted attractor. 

\subsubsection{The Hopf-bifurcation}
Our main aim for this  section is to show the following.
\begin{thm} 
There exists a $\gamma_0$ such that for $\frac 56+\frac{\sqrt{721}}{42}<\gamma<\gamma_0$  there exists a closed orbit, $c(\tau)$, acting as the future attractor for a set of non-zero measure of tilted Bianchi type VI$_{-1/9}$ models.
\end{thm}
To prove this theorem we will first show the existence of a closed period orbit 
acting as an attractor in a particular subset. We consider the invariant subset given by the Collinson-French solution with 2 tilts. We introduce $(X,Y)=(v_1,v_2^2)$:
\beq 
X'&=&\left(T-\frac 23\right)X-\frac{2\sqrt{6}}{3}Y, \nonumber \\
Y'&=&2\left(T+\frac {2\sqrt{6}}3X\right)Y. 
\label{2dsys}\eeq 
We set $X_0=-\frac{\sqrt{6}(3\gamma-4)}{2(3-\gamma)}$, 
$Y_0=\frac{5(3\gamma-4)(3-2\gamma)}{2(3-\gamma)^2}$, and perform the 
transformation $x=X-X_0$, $y=Y-Y_0$ with respect to the equilibrium point $\widetilde{\mathcal{C}\mathcal{F}}_{1+}$. 

The normal form of a Hopf-bifurcation can be written 
\beq Z'=(\lambda+b|Z|^2)Z, \eeq 
where $b$ is some complex number and $\lambda=\alpha+i \beta$ is a parameter. 
If $\mathrm{Re}(b)<0$ for $\alpha=0$, there exists a stable closed orbit for 
$0<\alpha$ sufficently small. 

We will therefore set $\gamma=\frac 56+\frac{\sqrt{721}}{42}$ and expand to cubic terms in $x$ and $y$. It is also convenient to introduce a complex variable $z$ chosen such that it aligns with the Jordan form of the linearised matrix. This can be achieved by setting:
\[ z=x+iay,\] 
where $a$ is a real number chosen such that the linear term of $z'=f(z,\bar{z})$ is 
\beq
\partial_zf(0,0)=\frac{i}{6}\sqrt{-1205+45\sqrt{721}}, \quad  \partial_{\bar{z}}f(0,0)=0.
\eeq
By a transformation
\[ Z=z+a_{11}z^2+a_{12}z\bar{z}+a_{22}\bar{z}^2+a_{111}z^3+a_{112}z^2\bar{z}+a_{122}z\bar{z}^2+a_{222}\bar{z}^3, \]
 we can choose the coefficients $a_{ij}$ and $a_{ijk}$ so that the equation for $Z$ takes the form 
\beq Z'=(\lambda+b|Z|^2)Z +\mathcal{O}(|Z|^4), 
\eeq 
where 
\[ \lambda=\frac{i}{6}\sqrt{-1205+45\sqrt{721}},\quad \mathrm{Re}(b)=\frac{1}{291600}\left(-15990233+595193\sqrt{721}\right)\approx -0.029.\]
Hence, there exists a $\gamma_0$ such that there exists a closed stable orbit for $\frac 56+\frac{\sqrt{721}}{42}<\gamma<\gamma_0$. 

The next step is to show that this orbit is also stable in the fully tilted type 
VI$_{-1/9}$ models. We can show this as follows: for a function $B$, we introduce 
\emph{the average}, $\av{B}$, with respect to the closed orbit $c(\tau)$, defined by
\[ \av{B}=\frac{1}{T}\oint_{c(\tau)}B\d \tau, \quad T=\oint_{c(\tau)}\d\tau.\] 
We can use this average, using similar manipulations as for the closed curves in the type IV, VI$_h$ and VII$_h$ models \cite{HHC,HHLC2} to show that: 
\begin{thm} 
Assume that there exists a closed periodic orbit $c(\tau)$ for the dynamical system (\ref{2dsys}). Then 
\[ \av{X}=-\frac{\sqrt{6}(3\gamma-4)}{2(3-\gamma)}, \quad \av{\lambda_{\Omega}}=-\frac{5(5\gamma-6)}{3(3-\gamma)}.\]
\end{thm}
\begin{proof}
From the $Y$ equation, we get $\av{T}=-2\sqrt{6}\av{X}/3$. A manipulation of the $V$ 
equation yields $\av{T}=\av{\mathcal{S}}$, and 
$\av{\mathcal{S}}=(3\gamma-4)-2(\gamma-1)A\av{X}$. 
These can now be solved to yield the desired value for $\av{X}$. 
A similar manipulation of the 
$\Omega$ equation yields $\av{\lambda_{\Omega}}$. 
\end{proof}
This implies that the vacuum solution is stable when the closed curve is 
perturbed by $\Omega$. Hence, since the Collinson-French solution is 
stable with respect to vacuum perturbations, this closed curve is stable for 
the fully tilted type VI$_{-1/9}$ models whenever  
$\frac 56+\frac{\sqrt{721}}{42}<\gamma<\gamma_0$. 

It remains to determine the maximal value for $\gamma_0$? It seems that this 
limiting value is related to the existence of a heteroclinic orbit originating and ending at the  
saddle points $\mathcal{C}\mathcal{F}_0$ and  $\widetilde{\mathcal{C}\mathcal{F}}_2$, respectively. 
For the limiting value of $\gamma$, which we will call $\gamma_H$, there exists 
such a heteroclinic orbit, while for values $\gamma\neq \gamma_H$ no such 
heteroclinic orbit exists connecting these two equilibrium points. 
We can use this to numerically estimate the value for $\gamma_H$. Our estimate gives:
\[ 1.473920<\gamma_H<1.473921.\] 
Since $5/6+\sqrt{721}/42=1.472653409...$ this means that the region in which the 
closed orbit exists is extremely small.\footnote{However, this seems to be 
typical for these models, the loophole for the type IV and VII$_h$ models also 
appears to be extremely small.}

Note also that this implies that these closed orbits co-exist with the extremely 
tilted attractor $\widetilde{\mathcal{E}\mathcal{C}\mathcal{F}}_-$, which makes 
this closed orbit even more difficult to detect. In addition, it appears as if the 
curves asymptote to the closed curve relatively slowly, which implies that the 
numerics have to run for a relatively long time in order to see the late-time asymptote. A numerical plot of some generally tilted type VI$_{-1/9}$ models approaching this closed orbit is shown in Fig.\ref{FigCC}.
\begin{figure}[f]
\caption{Bianchi type VI$_{-1/9}$ universes  approaching a closed curve and the equilibrium point $\widetilde{\mathcal{E}\mathcal{C}\mathcal{F}}_-$ ($\gamma=1.4735$). }\label{FigCC}
$$\includegraphics[scale=0.5,angle=90]{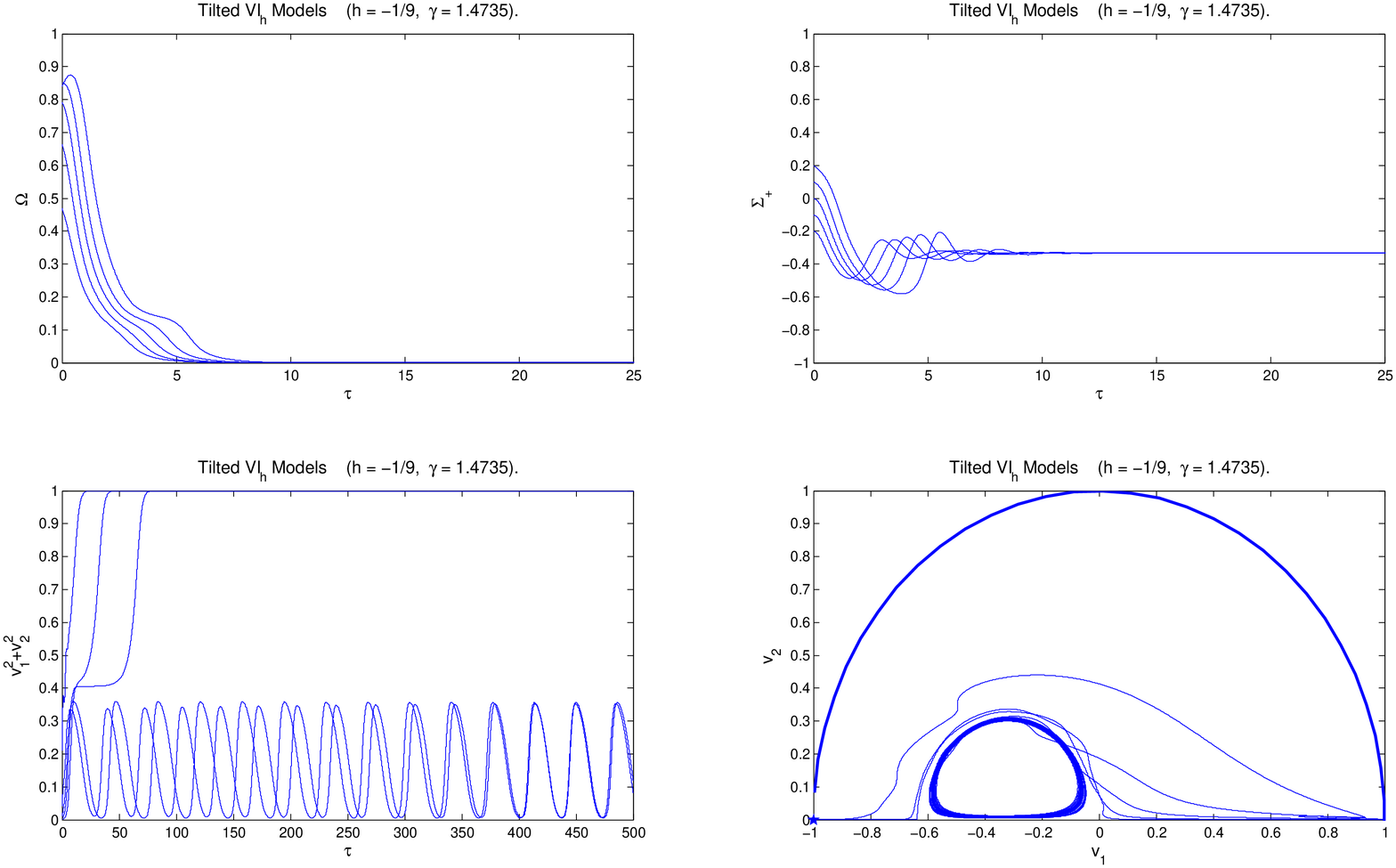}$$
\end{figure}

\section{Conclusion} 
In this paper we have examined the tilted Bianchi type VI$_{-1/9}$ model
in some detail.
This model is a special case (due to the vanishing of one of the constraint equations)
and necessitates a separate analysis from the general type
VI$_h$ models.  We showed that in these models there exists a tiny region of
parameter space where there exists a closed curve acting as an attractor.  This closed
curve co-exists with an extremely tilted attractor.
In the case of the most general irrotational models this closed curve, 
which appears in terms of the tilt
velocities, is absent.
We have also confirmed the
analytical results with an extensive numerical investigation.


\newpage
\begin{thebibliography}{99}

\bibitem{HBWII} C.G. Hewitt, R. Bridson, J. Wainwright, \textit{Gen.Rel.Grav.
} \textbf{33} (2001) 65

\bibitem{HWV} C.G. Hewitt and J. Wainwright, \textit{Phys. Rev.} \textbf{D46}
(1992) 4242


\bibitem{Harnett} D. Harnett, \textit{Tilted Bianchi type V cosmologies
with vorticity}, Master's thesis, University of Waterloo (1996).

\bibitem{BHtilted} J.D. Barrow and S. Hervik, \textit{Class. Quantum
  Grav.} \textbf{20} (2003) 2841

\bibitem{hervik} S. Hervik, \textit{Class. Quantum Grav.} \textbf{21} (2004) 2301

\bibitem{CH1} A. Coley and S. Hervik, \textit{Class. Quantum Grav.} \textbf{21} (2004) 4193-4208

\bibitem{CH2} A.A. Coley and S. Hervik,  \textit{Class. Quant. Grav.}
\textbf{22} (2005) 579. 

\bibitem{HHC} S. Hervik, R.J. van den Hoogen and A.A. Coley,
\textit{Class. Quant. Grav.} \textbf{22} (2005) 607. 

\bibitem{HHLC}  S. Hervik, R.J. van den Hoogen,
W.C. Lim and A.A. Coley, \textit{Class. Quant. Grav.} \textbf{23}
(2006) 845. 

\bibitem{HLim} S. Hervik and W.C. Lim, \textit{Class. Quantum Grav.} \textbf{23} (2006) 3017. 

\bibitem{HHLC2}  S. Hervik, R.J. van den Hoogen,
W.C. Lim and A.A. Coley, \textit{Class. Quant. Grav.}, \textbf{24} (2007) 3859 

\bibitem{DS1}  J. Wainwright and G.F.R. Ellis, \textit{Dynamical Systems
in Cosmology}, Cambridge University Press (1997)

\bibitem{DS2} A.A. Coley, \textit{Dynamical Systems
and Cosmology}, Kluwer, Academic Publishers (2003).


\bibitem{HHW-19}
C.G. Hewitt, J.T. Horwood and J. Wainwright,
  Class. Quant. Grav.  {\bf 20} (2003) 1743

\bibitem{CollinsEllis} 
C.B. Collins and G.F.R. Ellis, \textit{Phys. Rep.} \textbf{56} (1979) 65 

\bibitem{tilt1} A.A. Coley, S. Hervik and  W.C. Lim, \textit{Phys. Lett. B} \textbf{638} (2006) 310-313;
 A.A. Coley, S. Hervik and W.C. Lim,  \textit{Class. Quant. Grav.} \textbf{23} (2006) 3573-3591; 
A.A. Coley, S. Hervik and W.C. Lim,  \textit{Int. J. Mod. Phys.} \textbf{D15} (2006) 2187-2190;
 W.C. Lim, A.A. Coley and S. Hervik, , \textit{Class. Quant. Grav.} \textbf{24} (2007) 595-604

\bibitem{Apo2} P.S. Apostolopoulos,  
\textit{Class. Quantum Grav.} {\bf 22} (2005) 323-338
\end{thebibliography}
\end{document}